\documentclass{cccg23}
\usepackage{graphicx,amssymb,amsmath}
\usepackage{mathtools}
\usepackage{cite}
\usepackage{hyperref}
\hypersetup{colorlinks=true,linkcolor=[rgb]{0.75,0,0},citecolor=[rgb]{0,0,0.75}}

\title{ Analysis of Dynamic Voronoi Diagrams in the Hilbert Metric}
\author{Madeline Bumpus\thanks{Howard Universty, Washington DC, USA, {\tt gmadeline. bumpus@bison.howard.edu}}
        \and
        Xufeng Caesar Dai\thanks{Haverford College, Haverford, Pennsylvania, USA, {\tt  xdai1@haverford.edu}}
        \and
        Auguste H. Gezalyan\thanks{Department of Computer Science, University of Maryland, College Park, USA {\tt  octavo@umd.edu}}
                \and
        Sam Muñoz\thanks{Colby College, Waterville, Maine, USA, {\tt  smunoz23@colby.edu}}
                        \and
        Renita Santhoshkumar\thanks{Montgomery Blair High School, Silver Spring, Maryland, USA, {\tt  renitadsanthoshkumar@gmail.com}}
                \and
        Songyu Ye\thanks{Cornell Unversty, Ithaca, NY, USA, {\tt  sy459@cornell.edu}}
        \and
       David M. Mount\thanks{Department of Computer Science, University of Maryland, College Park, USA, {\tt  mount@umd.edu}}}

\index{Bumpus, Madeline}
\index{Dai, Xufeng}
\index{Gezalyan, Auguste}
\index{Muñoz, Sam}
\index{Santhoshkumar, Renita}
\index{Ye, Songyu}
\index{Mount, David}

\newcommand{\ang}[1]{\langle #1\rangle}

\renewcommand{\ang}[1]{\langle #1\rangle}
\newcommand{\RE}{\mathbb{R}}            

\newcommand{\eps}{\varepsilon}          
\newcommand{\ST}{\,:\,}                 

\newcommand{\SP}{\kern+1pt}             
\DeclareMathOperator{\interior}{int}
\DeclareMathOperator{\Vor}{Vor}
\newtheorem{definition}{Definition}
\newtheorem{observation}{Observation}

\begin{document}
\thispagestyle{empty}
\maketitle

\begin{abstract}

The objective of this paper is to study the properties of Hilbert bisectors and analyze Hilbert Voronoi diagrams in the dynamic setting. Additionally, we introduce dynamic visualization software for Voronoi diagrams in the Hilbert metric on user-specified convex polygons.
\end{abstract}

\section{Introduction}

In 1895, David Hilbert introduced the Hilbert metric~\cite{hilbert1895linie},
a projective metric used to measures distances within a convex body. It generalizes the Cayley-Klein model of hyperbolic geometry (on Euclidean balls) to any convex body. Hilbert geometry provides new insights into classical questions from convexity theory and the study of metric and differential geometries (such as Finsler geometries). 

In particular, Hilbert Geometry is of interest in the field of convex approximation. To approximate convex bodies, many techniques \cite{AFM17c,AAFM19,AFM17b,EHN11,RoV21,EiV21,NaV19,AAFM20} make use of covering convex bodies with regions that behave like metric balls in the Hilbert metric. These regions go by various names, such as: Macbeath regions, Macbeath ellipsoids, Dikin ellipsoids, and $(2,\eps)$-covers. In this way, a deeper understanding of the Hilbert metric can lead us to a deeper understanding of convex approximation. 

The Hilbert metric is applicable to many other fields of study such as quantum information theory \cite{reeb2011hilbertquantum}, real analysis \cite{lemmens2013birkhoff},  optimization \cite{chen2016entropic}, and network analysis \cite{chen2019relaxed}. For example, the Hilbert metric is crucial in nonlinear Frobenius-Perron theory \cite{lemmens2012nonlinear, chen2019relaxed}. In addition, the Hilbert metric has often been used in the context of optimal mass transport and Schr\"{o}dinger bridges \cite{chen2015optimal,chen2016entropic,chen2019relaxed, peyre2019computational}.  

There has been little work on the design and analysis of algorithms in the Hilbert geometry on convex polygons and polytopes, and most of it has been done in the last two decades. In this paper, we build off of Nielson and Shao~\cite{nielsen2017balls} as well as Gezalyan and Mount~\cite{gezalyan2021voronoi} to implement an algorithm for creating Voronoi diagrams in the Hilbert metric, characterize their bisectors as conics with an explicit formula, and analyze them in the dynamic setting. 

\section{Preliminaries}

\subsection{Previous Work} \label{sec:hilbert previous}

In 2017, Nielsen and Shao~\cite{nielsen2017balls} presented a characterization of the shape of balls in the Hilbert metric and explored their properties. They showed that the shape of a ball depends on the location its center as well as the vertices of the Hilbert domain. Nielsen and Shao gave an explicit description of Hilbert balls and studied the intersection of two Hilbert balls. In 2021, Gezalyan and Mount \cite{gezalyan2021voronoi} expanded on this by studying the geometric and combinatorial properties of Voronoi diagrams in the Hilbert metric. They presented two algorithms, a randomized incremental and a divide and conquer algorithm, for constructing these diagrams. We extend these works to analyze other properties of Voronoi diagrams in the Hilbert geometry, and we develop a software package implementing an incremental algorithm based on the randomized incremental algorithm. 
\subsection{Defining the Hilbert Metric} \label{sec:hilbert definitions}

A \emph{convex body} $\Omega \subset \mathbb{R}^d$ is a closed, compact, full-dimensional convex set. Given two fixed points $s, t \in \mathbb{R}^d$, let $\|s - t\|$ denote the Euclidean distance and $H(s,t)$ denote the Hilbert distance between them. Let $\chi(s,t)$ denote the \emph{chord} defined as the intersection of the line passing through $s$ and $t$ with $\Omega$.


\begin{definition} [Cross Ratio] \label{def:cross-ratio}
Given four distinct col\-linear points $a,b,c,d$ in $\mathbb{R}^d$, their cross ratio is defined to be:
\begin{align*} 
    (a,b;c,d)
        ~ = ~ \frac{\|a-c\|\|b-d\|}{\|b-c\|\|a-d\|}.
\end{align*}
Where the orientation of the line determines the sign of each distance.
\end{definition}
Note that the cross ratio is invariant under projective transformations \cite{papadopoulos2014funkarxiv}. 


\begin{definition}[Hilbert metric] \label{def:hilbert-metric}
Given a convex body $\Omega$ in $\mathbb{R}^d$ and two distinct points $s, t \in \interior(\Omega)$, let $x$ and $y$ denote endpoints of the chord $\chi(s,t)$, so that the points are in the order $\ang{x, s, t, y}$. The Hilbert distance between $s$ and $t$ is defined as:


\begin{align*}
    H_{\Omega}(s,t)
        ~ = ~ \begin{cases}
            \frac{1}{2}\ln (s,t;y,x) & \text{\quad if $s\neq t$} \\
            0 & \text{\quad if $s=t$}.
        \end{cases} 
\end{align*}
    
\end{definition}
 It is well known that this is a metric, and in particular, it is symmetric and satisfies the triangle inequality. Moreover, it is well known that $ H(a,b) + H(b,c) = H(a,c)$ if and only if $r_\Omega(a,b),r_\Omega(b,c),r_\Omega(a,c)$ are collinear and $r_\Omega(b,a),r_\Omega(c,b),r_\Omega(c,a)$ are collinear, where $r_\Omega(a,b)$ is the point where the ray $ab$ meets $\partial(\Omega)$ \cite[Theorem 12.4]{papadopoulos2014funkarxiv}. This follows from the Hilbert metric's characterization as being the average of the forward and backwards Funk metrics.
It is also well known that geodesics are unique if and only if $\Omega$ is strictly convex \cite[Corollary 12.7]{papadopoulos2014funkarxiv}. 

\subsection{Voronoi Diagrams in the Hilbert Metric} \label{sec:hilbert geodesics}
Let $\Omega\subset \mathbb{R}^2$ be an open convex region whose boundary, denoted $\partial \Omega$, is a polygon. Let $S$ denote a set of $n$ sites in $\Omega$. The Voronoi cell of a site $s\in S$ is:
\begin{align*}
    V(s)
        ~ = ~  \big\{ q \in \Omega \ST d(q,s) \leq d(q,t), \,\forall t \in S \setminus \{s\} \big\}.
\end{align*}
 The \emph{Voronoi diagram} of $S$ in the Hilbert metric induced by $\Omega$, denoted $\Vor_\Omega(S)$, is the cell complex of $\Omega$ induced by the Voronoi cells $V(s)$, for all $s \in S$. Outside of the dynamic context, we assume that all sites are in general position, and in particular that no three sites are co-linear. It is known that Voronoi cells in the Hilbert Metric are stars and with complexity $\Omega(mn)$\cite{gezalyan2021voronoi}. 

It will be helpful to define two terms, \emph{spokes} and \emph{sectors}. Given our convex body $\Omega$ and a site $s$ in $\Omega$ it can be seen that the Hilbert distance from our site $s$ to a point $p$ is dependent on the edges of $\Omega$ that intersect $\chi(s,p)$. This is characterized by taking the vertices of $\Omega$, which we denote by $V$, and partitioning $\Omega$ into cells determined by the chords $\chi(v,s)$ for all $v \in V$. We define a \emph{spoke} $r_\Omega(v,s)$, $r_\Omega(s,v)$ to be the intersection of ray $vs$ and $sv$ respectively with $\Omega$ when $v$ is a vertex of $V$. Nielsen and Shao showed that Hilbert balls are $\Theta(m)$-sided polygons whose boundaries are constructed around these \emph{spokes} \cite{nielsen2017balls}.

These \emph{spokes} subdivide $\Omega$ into polygonal regions, which we will call \emph{sectors}. Each sector $T$ has the property that for all points $p\in T$, given two sites $s$ and $t$ there are four edges (not necessarily unique) of $\partial \Omega$, say $E_A,E_B,E_C,E_D$, such that $\chi(s,p)$ always intersects $\Omega$ at $E_A$ and $E_B$ and $\chi(t,p)$ at $E_C$ and $E_D$. We will write $S(s,t,E_A,E_B,E_C,E_D)$ to be the unique \emph{sector} such that for any $p \in S(s,t,E_A,E_B,E_C,E_D)$, we have $\chi(sp)$ has in order of intersection $E_A, s, p, E_B$ and $\chi(tp)$ has $E_C, t, p, E_D$ (see Figure~\ref{fig:Sector}). As a consequence of this, we can see that the bisector between two sites, $s$ and $t$, is composed of curves which depend piece-wise on the sectors they pass through.

\begin{figure}[t] 
    \centerline{\includegraphics[scale=0.35]{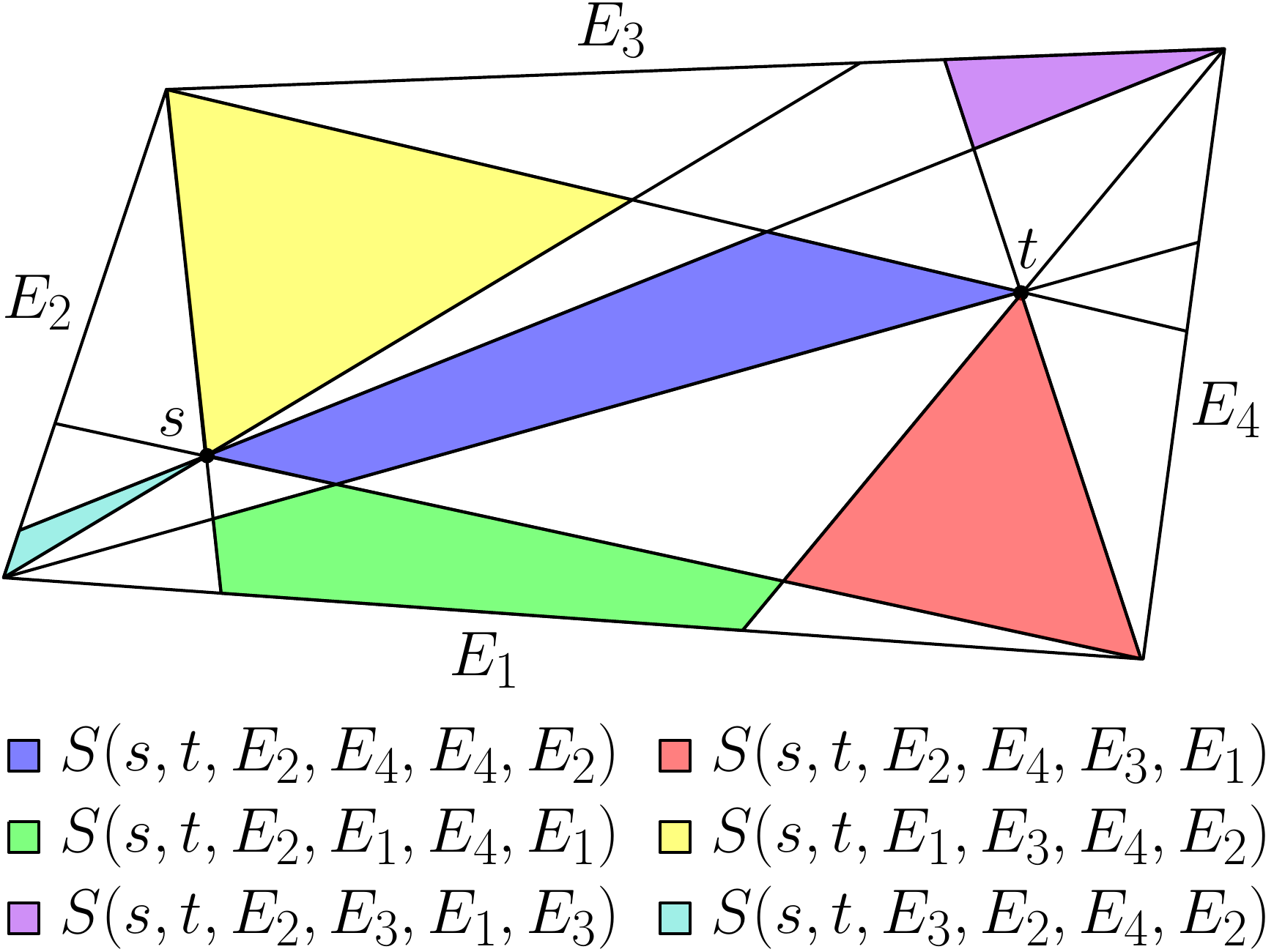}}
    \caption{Illustration of different \emph{sectors}.}\label{fig:Sector}
\end{figure}


\section{Hilbert Bisector analysis}

\begin{theorem}\label{GeneralEquation}
The equation of the bisector between two sites in any particular sector is of the form $Ax^2+Bxy+Cy^2+Dx+Ey+F = 0$, with coefficients depending only on the line equations of the four edges and the sites. 
\end{theorem}
\begin{proof} 
Suppose we are in some arbitrary sector $S(s,t,E_A,E_B,E_C,E_D)$. Let the equation of the edges be $a_1x + a_2y + a_3 = 0$,  $b_1x + b_2y + b_3 = 0$, $c_1x + c_2y + c_3 = 0$, and $d_1x + d_2y + d_3 = 0$, for $E_A$, $E_B$, $E_C$, and $E_D$ respectively. 

Recall the Euclidean distance between a point $(x,y)$ and a line $ux+vy+l=0$ is given by $D(u,v,l,x,y) = \frac{|ux + vy + l|}{\sqrt{u^2 + v^2}}$. Let $p\in T$ be an arbitrary point with coordinates $(p_x,p_y)$. Notice that, without loss of generality, if $q$ is the point where chord $\chi(p,s)$ hits $\partial\Omega$ towards $s$, then the ratio $\frac{\|q-p\|}{\|q-s\|}$ can be replaced with the ratio $\frac{D(a_1,a_2,a_3,p_x,p_y)}{D(a_1,a_2,a_3,s_x,s_y)}$ by similar triangles. Using this technique we set $H(s,p)=H(t,p)$ and solve, yielding
\begin{align*}
    \frac{D(a_1,a_2,a_3,p_x,p_y)}{D(a_1,a_2,a_3,s_x,s_y)} \frac{D(b_1,b_2,b_3,s_x,s_y)}{D(b_1,b_2,b_3,p_x,p_y)} 
        ~ = ~ \\ \frac{D(c_1,c_2,c_3,p_x,p_y)}{D(c_1,c_2,c_3,t_x,t_y)} \frac{D(d_1,d_2,d_3,t_x,t_y)}{D(d_1,d_2,d_3,p_x,p_y)}.
\end{align*}

 The equation of the line $ux+vy+l=0$ divides the plane in two, characterized by the sign of the function $g(x,y)=ux + vy + l$. Given this, since all our points are on the same side of each line, when we substitute in for the function $D$, we can remove the absolute value bars from our equations. It follows that: 
\begin{align*}
    \frac{a_{1} p_x+a_{2} p_y+a_{3}}{a_{1} s_x+a_{2} s_y+a_{3}} \frac{b_{1} s_x+b_{2} s_y+b_{3}}{b_{1}t p_x+b_{2} p_y+b_{3}}
        ~ = ~ \\ \frac{c_{1} p_x+c_{2} p_y+c_{3}}{c_{1}t_x+c_{2}t_y+c_{3}} \frac{d_{1} t_x+d_{2} t_y+d_{3}}{d_{1} p_x+d_{2} p_y+d_{3}}.
\end{align*}
Collecting constants on one side gives us:
\begin{align*}
    \frac{b_{1}s_x+b_{2}s_y+b_{3}}{d_{1}t_x+d_{2}t_y+d_{3}}\frac{c_{1}t_x+c_{2}t_y+c_{3}}{a_{1}s_x+a_{2}s_y+a_{3}}
        ~ = ~ \\ \frac{c_{1}p_x+c_{2}p_y+c_{3}}{a_{1}p_x+a_{2}p_y+a_{3}}\frac{b_{1}p_x+b_{2}p_y+b_{3}}{d_{1}p_x+d_{2}p_y+d_{3}}.
\end{align*}
Letting the left side be a constant $k$, we find that the bisector satisfies the algebraic equation:
\begin{align*}
    Ap_x^{2} + Bp_xp_y +Cp_y^{2}+Dp_x+Ep_y+ F ~ = ~ 0,
\end{align*}
where the coefficients are: 
\begin{align*}
    A &= b_{1}c_{1}-a_{1}d_{1}k \\
    B &= b_{2}c_{1}+b_{1}c_{2}-a_{1}d_{2}k-a_{2}d_{1}k\\
    C &= b_{2}c_{2}-a_{2}d_{2}k\\
    D &= b_3c_1+c_3b_1-a_3d_1k-a_1d_3k\\
    E &= b_3c_2+b_2c_3-a_2d_3k-a_3d_2k\\
    F &= b_3c_3-a_3d_3k.
\end{align*}

By the symmetry of these equations we can see that the bisector in $S(s,t,E_A,E_B,E_C,E_D)$ is the same as the bisector in $S(s,t, E_B,E_A,E_D,E_C)$.\footnote{We would like to thank Daniel Skora for comments on a previous version of this proof.}
\end{proof}

We present several ways to simplify the bisector depending on the sector case. 

\subsection{Finding Bisectors in a Sector}
There are three main sector types: sectors with four distinct edges, sectors with three, and sectors with two. We begin with the four edge case. It is a well known fact from projective geometry that there exists a projective transformation that maps any convex quadrilateral to another (see Figure \ref{fig:ProjectionQuadrilaterals}). For the sake of completeness, we present the derivation of this transformation for the special case where the target is the unit square.

\begin{observation}\label{4EdgeProjection}
Bisectors in the four-edge case can be projectively transformed and solved in the unit square. 
\end{observation}

\begin{proof}
Suppose we are in a strictly four-edge sector $S(s,t,E_A,E_B,E_C,E_D)$ with edges $E_A,E_C,E_B,E_D$ counter clockwise. Let our edges extend out infinity and let $p_1=E_A \cap E_D$, $p_2=E_A \cap E_C$, $p_3=E_B \cap E_C$, and $p_4=E_B \cap E_D$. Consider the vector:
\[
    Q^T ~ = ~ [q_{1,x},q_{1,y},q_{2,x},q_{2,y},q_{3,x},q_{3,y},q_{4,x},q_{4,y}],
\]
where $q_1=(0,0),$ $q_2=(0,1),$ $q_3=(1,1),$ and $q_4=(1,0)$. Then we can calculate the projective transformation matrix $T$ such that $P'=TP$ when $T$ sends points $P$ on our quadrilateral to the unit square. We force the right corner of $T$ to be 1 to fix our scaling factor.
\begin{figure}[t]
    \centerline{\includegraphics[scale=0.5]{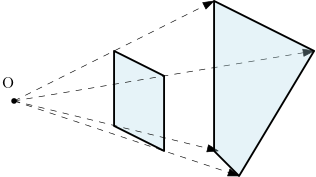}}
    \caption{A projection between two quadrilaterals.}\label{fig:ProjectionQuadrilaterals}
\end{figure} 
\[  \left[ {\begin{array}{c}
    x' \\
    y' \\
    1  \\
  \end{array} } \right]
  \left[ {\begin{array}{ccc}
    t_{11} & t_{12} & t_{13} \\
    t_{21} & t_{22} & t_{23} \\
    t_{31} & t_{32} & 1 \\
  \end{array} } \right]
=\left[ {\begin{array}{c}
    x \\
    y \\
    1  \\
  \end{array} } \right]\]

We solve for the elements of T using the following matrices:
$M V = Q$ where $V = [t_{11}, t_{12}, \ldots]^\intercal$, $Q = [q_{1,x}, q_{1,y}, \ldots]^\intercal$ and $M$ is:
\[
  \left[ \kern-3pt \begin{array}{cccccccc}
    p_{1,x} & p_{1,y} & 1 & 0 & 0 & 0 & -p_{1,x} \SP q_{1,x} & -p_{1,y} \SP q_{1,x}\\
    0 & 0 & 0 & p_{1,x} & p_{1,y} & 1 & -p_{1,x} \SP q_{1,y} & -p_{1,y} \SP q_{1,y}\\
    p_{2,x} & p_{2,y} & 1 & 0 & 0 & 0 & -p_{2,x} \SP q_{2,x} & -p_{2,y} \SP q_{2,x}\\
    0 & 0 & 0 & p_{2,x} & p_{2,y} & 1 & -p_{2,x} \SP q_{2,y} & -p_{2,y} \SP q_{2,y}\\
    p_{3,x} & p_{3,y} & 1 & 0 & 0 & 0 & -p_{3,x} \SP q_{3,x} & -p_{3,y} \SP q_{3,x}\\
    0 & 0 & 0 & p_{3,x} & p_{3,y} & 1 & -p_{3,x} \SP q_{3,y} & -p_{3,y} \SP q_{3,y}\\
    p_{4,x} & p_{4,y} & 1 & 0 & 0 & 0 & -p_{4,x} \SP q_{4,x} & -p_{4,y} \SP q_{4,x}\\
    0 & 0 & 0 & p_{4,x} & p_{4,y} & 1 & -p_{4,x} \SP q_{4,y} & -p_{4,y} \SP q_{4,y}\\
  \end{array} \kern-3pt \right]
\]
Since the Hilbert metric is invariant under projective transformations, the final bisector can thus be computed for the canonical case of the unit square and then mapped back using the inverse transformation.
\end{proof}

Given this observation, we will assume henceforth that any four edges case is in the unit square. 

\begin{lemma}\label{4EdgeHyperbola}
When the four edges corresponding to a sector lie on four distinct lines, the bisector is either a hyperbola or a parabola.
\end{lemma}

\begin{proof}
Let $S(s,t,E_A,E_B,E_C,E_D)$ be an arbitrary sector with four distinct edges. By applying Observation~\ref{4EdgeProjection}, we may assume without loss of generality that $E_A$, $E_B$, $E_C$, $E_D$ lie on the lines $x=0$, $x=1$, $y=0$, and $y=1$, respectively. The discriminant of equation for our bisector is given by $B^2-4AC$. Substituting in for our variables, we get $(1-k)^2$ as our discriminant, which is nonnegative. 
\end{proof}

Next, let us consider the three-edge case. We will refer to the edge that is hit twice by the Hilbert metric in the sector as the \emph{shared edge}, we have the following three cases, depending on whether the shared edge is behind both points with respect to the sector, in front of both (see Figure \ref{fig:threeEdgeInFront}), or behind one and in front of the other (see Figure \ref{fig:3edgeAllConics}). As in the four-edge case, we can projectively transform the problem to a canonical form, as given in the following lemma. Define the \emph{unit simplex} in $\RE^2$ to be the right triangle with vertices at $(0,0)$, $(1,0)$, and $(0,1)$.

\begin{observation}\label{3edgeProjection}
Bisectors in the three-edge case can be projectively transformed and solved in the unit simplex.
\end{observation}

\begin{proof}
The affine transformation between triangles is well known and inherently projective. We provide it here. A triangle with corners $p,q,r$ can be projectively transformed into the unit simplex with the matrix $T^{-1}$, where:  
\[
  T=\left[ {\begin{array}{cccccccc}
    p_x-r_x & p_y-r_y & 0\\
    q_x-r_x & q_y-r_y & 0\\
    r_x &r_y & 1\\

  \end{array} } \right].
\]
Observe that $T^{-1}$ sends a point $P$ on the left from our original triangle to the unit simplex.
\end{proof}

\begin{lemma}\label{3edgeLine}
Given any three-edge sector, where the shared edge is in front or behind of both sites, the bisector between the two sites in the sector is a straight line through the vanishing point of the non-shared edges.
\end{lemma}
\begin{figure}[h]
    \centerline{\includegraphics[scale=0.5]{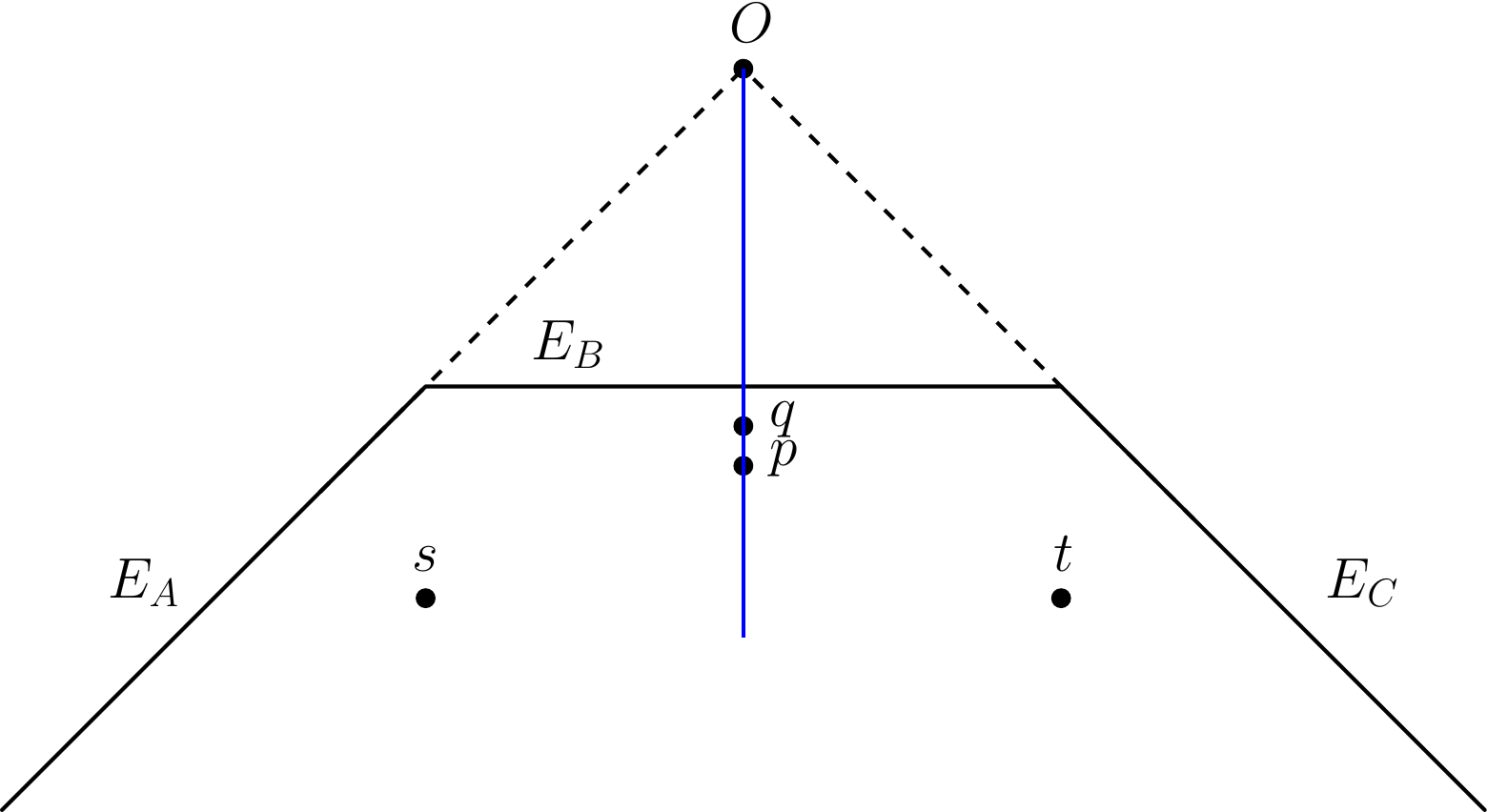}}
    \caption{Illustration for the proof of Lemma \ref{lem:3edgeAllConics} (a bisector in the three-edge case with shared edge in front of both sites).}
    \label{fig:threeEdgeInFront}
\end{figure}
\begin{proof}  We prove this geometrically for both cases $S(s,t,E_A,E_B,E_C,E_D)$ with edges $E_B=E_D$ and $S(s,t,E_A,E_B,E_C,E_D)$ with edges $E_A=E_C$.

Consider the case in which the edge is in front of both sites, $S(s,t,E_A,E_B,E_C,E_B)$. Let $p$ be a point in the sector such that $H(s,p)=H(t,p)$, and let $O$ be vanishing point of $E_A$ and $E_C$. We claim the bisector is a straight line through $Op$. Let $q$ be a point on the segment $Op$ in the sector. Using the geodesic triangle inequality in the Funk metric we get \(F(t,q) = F(t,p) + F(p,q)\) and \(F(s,q) = F(s,p) + F(p,q)\)~\cite{papadopoulos2014funkarxiv}. By definition, $F(s,p)+F(p,s)=F(t,p)+F(p,t)$. Substituting in our previous equations and cancelling $F(p,q)$ we have $F(s,q)+F(p,s)=F(t,q)+F(p,t)$. Using properties of similar triangles and rearranging the equation we can see that  $F(s,q)-F(t,q)=\log((|pE_c|/|tE_c|)\cdot(|sE_A|/|pE_A|))$. Note that $|pE_c|/|pE_A|=|qE_c|/|qE_A|$. Substituting in and using the definition of the Funk metric we get $F(s,q)-F(t,q)=F(q,t)-F(q,s)$, yielding $H(s,q)=H(t,q)$. By choosing $p$ to be the lowest point on the bisector in the sector we get our characterization of the bisector as a line. 

Note the case where the shared edge is behind both sites follows from the symmetry of the general equation of the bisector. 
\end{proof}

\begin{lemma} \label{lem:3edgeAllConics}
Given a three-edge sector in which the shared edge is in front of one site and behind the other, the bisector between the sites can be a conic of any type (ellipse, parabola, or hyperbola) depending on their relative positions. Further, when either site is fixed, there exists a line such that the type of conic is determined by the location of the other site with respect to this line (an ellipse if it lies on one side, hyperbola if on the other side, and parabola if lies on the line). 
\end{lemma}

\begin{proof} We prove this algebraically for, without loss of generality, $S(s,t,E_A,E_B,E_C,E_D)$ with edges $E_B=E_C$, and all other edges lying on distinct lines.

By Observation~\ref{3edgeProjection}, it suffices to assume that the triangle has been mapped to the unit simplex. Let the edge $E_A$ lie on the line $y=0$, $E_B$ and $E_C$ lie on the line $x+y-1=0$, and $E_D$ lie on the line $x=0$. By applying Theorem \ref{GeneralEquation} with these values, the bisector is given by $x^2 + (2-k)xy + y^2 -2x -2y +1 = 0$, where $k=\frac{s_x+s_y-1}{t_x}\frac{t_x+t_y-1}{s_y}$. The discriminant of the equation of the bisector is $(2-k)^2-4$, which reduces to $k(k-4)$. 
By standard results on conics, the type of conic depends on the sign of the discriminant. Note that $k$ is always positive since $(s_x,s_y)$ and $(t_x,t_y)$ are in the unit simplex. Given this, the discriminant changes sign when $k = 4$. Thus, the conic type is given by the sign of the function $(s_x+s_y-1)(t_x+t_y-1) - 4 s_x t_x$. Observe that when either $s$ or $t$ is fixed, this is a linear equation in coordinates of the other site. Thus, fixing $s$ yields a line, and the type of conic (ellipse, hyperbola, or parabola) is determined by the location of $t$ relative to this line. The elliptical case holds when $t$ is on one side, the hyperbolic case when it is on the other, and the parabolic case when it lies on the line. This applies symmetrically for $t$.
\end{proof}

Figure~\ref{fig:3edgeAllConics} illustrates four examples of conics as described by Lemma~\ref{lem:3edgeAllConics}. The site coordinates and discriminant are given in Table~\ref{tab:3edgeAllConics}. Note that the ellipse case and hyperbola case happen in general position.

\begin{figure}[htbp]
    \centerline{\includegraphics[scale=0.165]{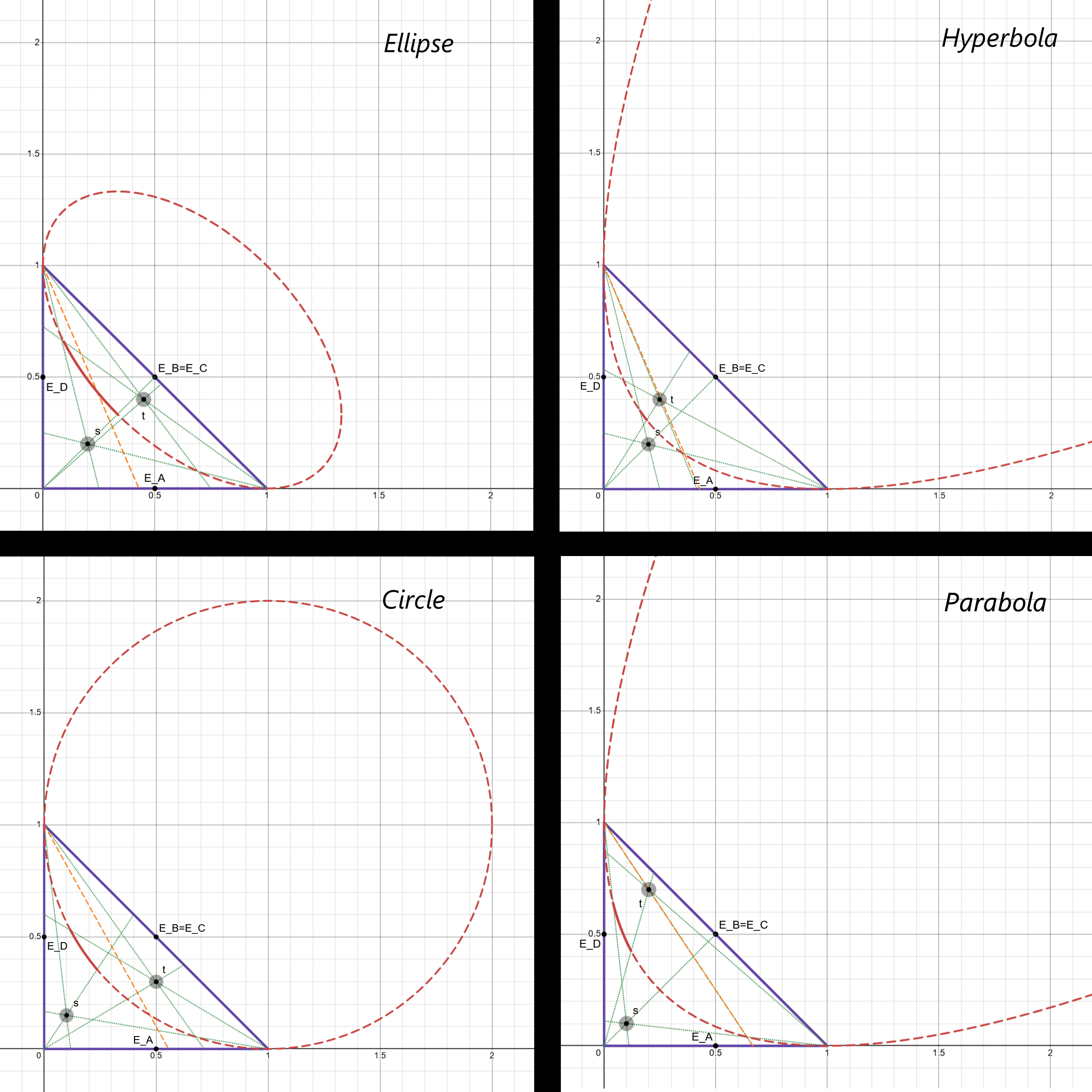}}
    \caption{The bisector conics for the cases given in Table~\ref{tab:3edgeAllConics} visualized.} \label{fig:3edgeAllConics}
\end{figure} 

\begin{table}[ht]
\caption{The bisector cases illustrated in Figure~\ref{fig:3edgeAllConics} for the three-edge case.} \label{tab:3edgeAllConics}
\begin{center}
\begin{tabular}{ |c|c|c|c|c| } 
 \hline
   & $s$ & $t$ & Discriminant \\ 
 \hline
 Ellipse & (0.2,0.2) & (0.45,0.4) & -3 \\
 \hline
 Hyperbola & (0.2,0.2) & (0.25,0.4) & 0.84 \\ 
 \hline
 Circle & (0.1,0.15) & (0.5,0.3) & -4 \\
 \hline
 Parabola & (0.1,0.1) & (0.2,0.7) & 0 \\
 \hline
\end{tabular}
\end{center}
\end{table}

Now that we have evaluated the three-edge case we will evaluate the two-edge case.

Note that given a sector $S(s, t, E_A, E_B, E_C, E_D)$ with $E_A=E_D$ and $E_B=E_C$, we can choose a point $p$ on $E_A$, $q$ on $E_B$, and let $E_A\cap E_B = r$ and affinely map the triangle $pqr$ into the unit simplex using Observation \ref{3edgeProjection}.

\begin{lemma}
Given a two edge sector, $S(s, t, E_A, E_B,$ $E_C, E_D)$ with $x=0$ as $E_A=E_D$, and $y=0$ as $E_B=E_C$, the bisector between $s$ and $t$ is given by the following linear equation: $y=\sqrt{k}x$ which passes through the vanishing point of $E_A$ and $E_B$ where $k = (s_y/t_x)\cdot(t_y/s_x)$.
\end{lemma}

\begin{proof} 

Plugging in to the general equation for the bisector in a sector give us $-kx^2+y^2=0$ with $k=(s_y/t_x)\cdot(t_y/s_x)$ and solving gives us one viable solution $y=\sqrt{k}x$. It is clear that this goes through the vanishing point of $E_A$ and $E_B$.
\end{proof}

\section{Analysis of Hilbert Voronoi Diagrams}

The Hilbert metric has some unique features that are not present in the Euclidean metric, In particular, in degenerate cases bisectors in the Hilbert metric can contain two dimensional polygons. Next, we present a necessary and sufficient condition for this occurrence. 

\begin{figure}[t]
    \centerline{\includegraphics[scale=0.55]{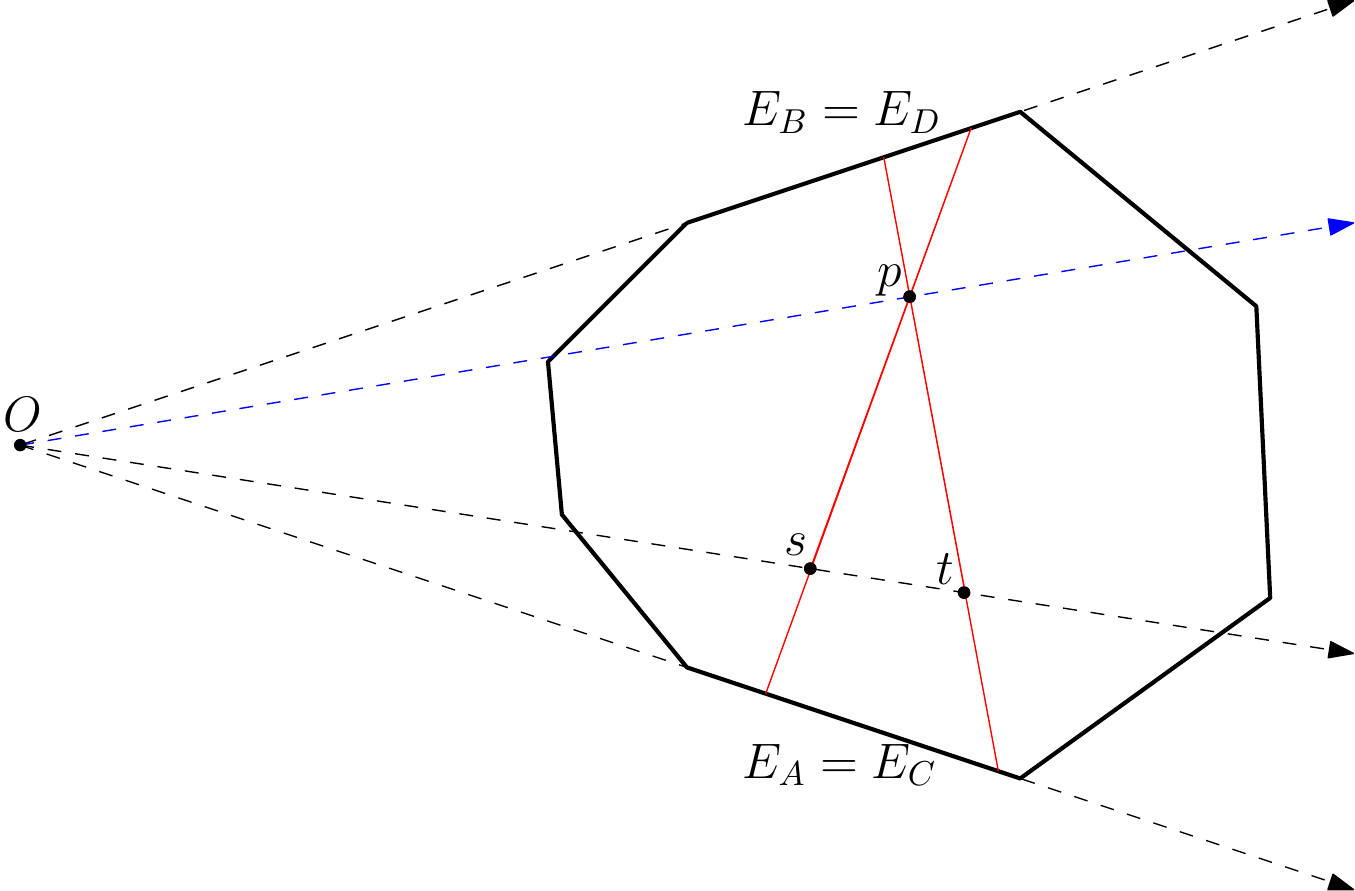}}
    \caption{The point $p$ is equidistant from both sites $s$ and $t$ by the invariance of the cross ratio.}\label{fig:2DBisectors}
\end{figure}

\begin{lemma}\label{lem:2DBisectors}  The bisector between two sites $s$ and $t$ contains a Hilbert ball of some radius $r$ if and only if both sites lie on a ray through a vertex $V$ at the vanishing point of two edges $E_A$ and $E_B$ of $\Omega$ and the sector $S(s,t,E_A,E_B,E_B,E_A)$ exists.
\end{lemma}

\begin{proof}
($\Longleftarrow$) Consider that the Hilbert distance from $s$ and $t$ to any point $p$ in $S(s,t,E_A,E_B,E_B,E_A)$, from the perspective of both $s$ and $t$, is equal by the invariance of the cross ratio. Hence, the entire sector in which both $s$ and $t$ are between $E_A$ and $E_B$ is part of the bisector between $s$ and $t$. By assumption, $\Omega$ is a bounded convex polygonal body, so our bisector will contain some fully 2-dimensional region. Note that we can always fit a Euclidean ball in any 2-dimensional region, and by we therefore must be able to fit a Hilbert ball in one as well  \cite{papadopoulos2014funkarxiv}.

($\Longrightarrow$) Suppose the bisector between two sites $s$ and $t$ contains a Hilbert ball $N$. Then there exists a ball $N'$ nested in $N$ so that $N'$ lies entirely within one sector $T$. Note that the cross ratio between either site and any point in $N'$ must be equivalent. We assert that $T$ is defined by at most two edges of $\Omega$.

Suppose that there were more than two unique edges defining $T$. Consider a point, $p$, in $N'$. Take a pair of edges with respect to one site and consider the ray from the vanishing point of that pair of edges through $p$. Note that moving $p$ along this line will not change its distance to the site, however, since this sector is defined by more than two edges, it will from the other site. Hence, there can be only two edges of $\Omega$ defining $T$.

By the continuity of the Hilbert Metric, it must be that both sites lie on the same side of the bisector in T. Otherwise, we could move closer to one site and further from the other. Since both sites lie on the same side of their bisector, share two edges with respect to every point in $N'$, and are the same distance to every point in $N'$, they must lie on a ray through $o$.
\end{proof}

Given this, we can see that any dynamic Hilbert Voronoi diagram will contain discontinuities whenever a site passes over a ray from a vanishing point of a pair of edges and another site. (See Figure~\ref{fig:discontinuties} for an example.)

\begin{figure}[t]
    \centerline{\includegraphics[scale=0.25]{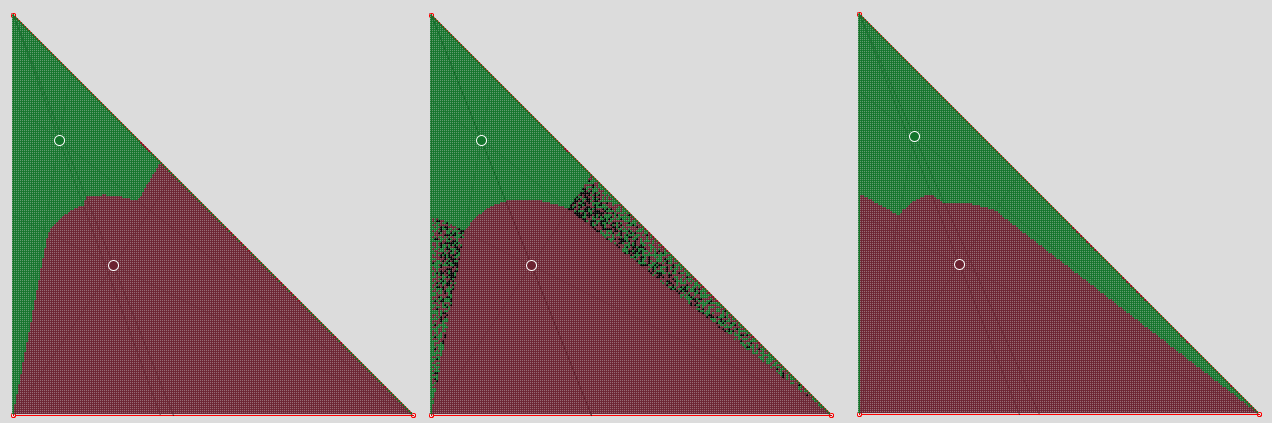}}
    \caption{A discontinuity in the Hilbert metric Voronoi.}\label{fig:discontinuties}
\end{figure}

Another property of Hilbert Voronoi diagrams that is distinct from Euclidean Voronoi diagrams is that three sites in general position do not necessarily create a Voronoi vertex. In particular given two site $s$ and $t$, there exists a space of positive measure such no ball containing a point in that space on its boundary can contain both $s$ and $t$ on its boundary.

\begin{lemma}\label{lem: voronoi vertex} There exists a quadrilateral region $Z$ of positive measure between any two sites $s$ and $t$ in a convex polygonal $\Omega$, such that for all $r \in Z$, there is no Hilbert ball whose boundary passes through $r$, $s$ and $t$.
\end{lemma}

\begin{proof}
Let $\Omega$ have $m$ edges, denoted $\{E_1, \ldots, E_m\}$.
Let $s$ and $t$ be two sites in $\Omega$. Take $O$ to be the set of vanishing points on pairs of edges $O= \{O_{1, 2}, O_{1, 3}, \dots, O_{m-1, m}\}$ when $O_{i,j}$ refers to the vanishing point of edges $E_i$ and $E_j$. Let $R_{i,j}$ be the area between rays $O_{i,j}s$ and $O_{i,j}t$ intersected with $\Omega$. Let $R=\{R_{1, 2}, R_{1, 3}, \dots, R_{m-1, m}\}$.

Let $Z=\bigcap_R R_{i,j}$. By definition, $Z$ is a polygon whose edges pass through both \textit{s} and \textit{t}, which always contains \textit{st}. Note that it is equal to \textit{st} if and only if \textit{s} and \textit{t} lie on a ray through an element of \textit{O}.

Consider the boundary of \textit{Z}.  We begin by proving that $Z$ is a quadrilateral.  
Without loss of generality let us consider only the top half of $Z$ above the line $st$. Let the first two rays leaving $s$ and $t$ and forming part of the boundary of $Z$ be denoted $l_s$ and $l_t$, respectively. Let $l_s$ and $l_t$ intersect at a point $v$. If $v$ is not in $Z$, there must exist some line $l$ that intersects both $l_s$ and $l_t$ and either $s$ or $t$, and therefore, must be one of the original rays $l_s$ or $l_t$. This contradicts the fact that $v$ is not in $Z$. Hence we know $Z$ must be a quadrilateral. Next, we prove that no ball through $s$ and $t$ can intersect the interior of $Z$ on its boundary.

 To do this, we prove that any ball containing $s$ and $t$ must contain $Z$. Consider an arbitrary Hilbert ball $N$ with $s$ and $t$ on its boundary. Note that the boundary of $N$ must be made of segments along rays passing through elements of $O$. 
Since $s$ and $t$ are on the boundary $N$, it must be that if an edge, $W$ of $N$, intersects $Z$ it would cut through, without loss of generality, the top half of the quadrilateral. Extending $W$, we see that it must intersect $l_s$ and $l_t$ at exactly one point each. Let $W$ come from a ray passing through $O'\in O$. Consider segments through $s$ and $t$ created by the rays through $O'$. Note that $W$ must be either above or below $O's$ and $O't$, or otherwise either $s$ or $t$ would not be on the boundary. However, $Z$ is contained in the region between $O's$ and $O't$, so $W$ cannot cut through it, contradiction. Hence we can see that $N$ must contain the interior of $Z$.
\end{proof}

\begin{figure}[t]    
    \centerline{\includegraphics[scale=0.55]{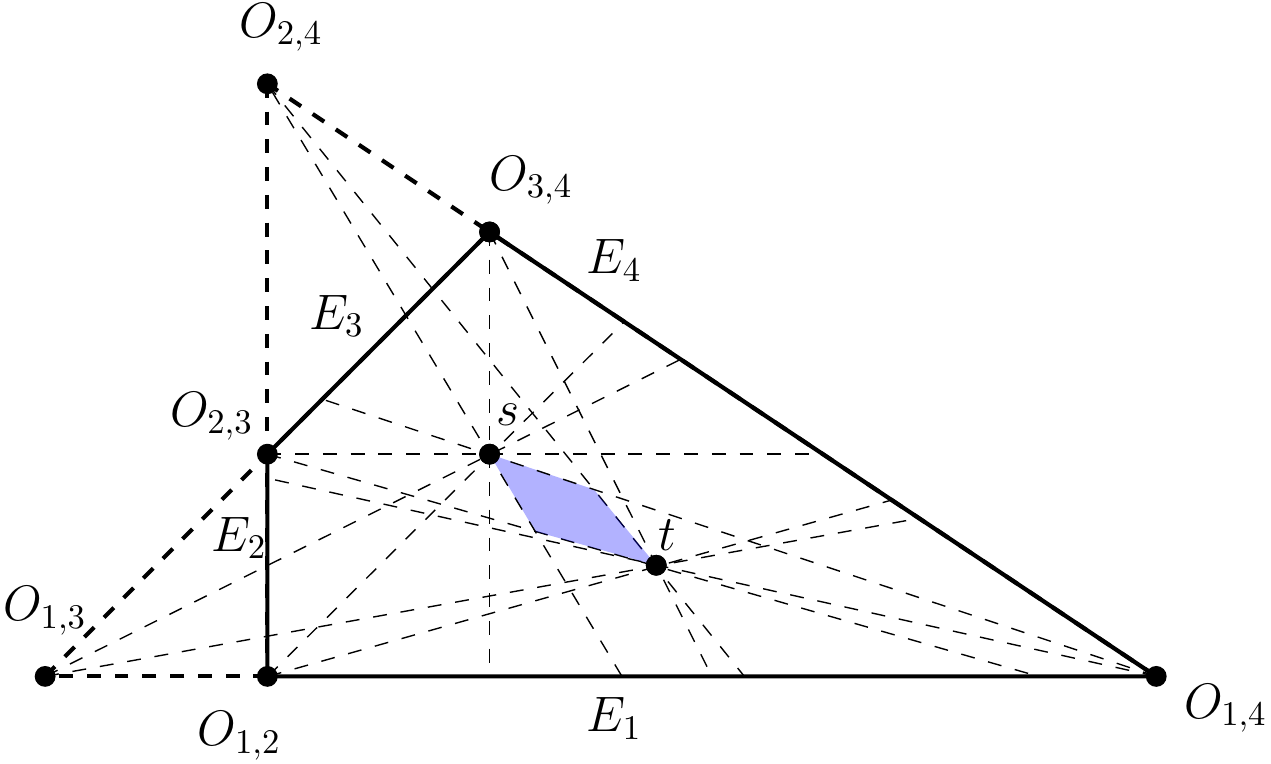}}
    \caption{Illustration of region Z.}\label{fig:t-zone}
\end{figure}
\section{Software}
In this section we present our software for Dynamic Voronoi Diagrams in the Hilbert metric.

Our software is built on the structure provided by the software presented by Nielsen and Shao~\cite{nielsen2017balls}. We contributed by implementing a dynamic insertion algorithm. The insertion process works as follows.

Given an existing convex body $\Omega$ with a set of sites $S$, suppose the user chooses to insert a new site $t$. Our software will subdivide $\Omega$ into sectors from the perspective of $t$. Then for each site $s \in S$ our software calculates $\Vor_{\Omega,\{s,t\}}(t)$. This is done by tracing the bisector through the cell decomposition on $\Omega$ induced by the sectors of both $s$ and $t$ using the method described in \cite{gezalyan2021voronoi}. While the bisector is approximated using line segments, the equation of the bisector is printed to the terminal window along with the sector information. Our algorithm will then update the Voronoi cell of all previously calculated $s\in S$'s as $\Vor_{\Omega,S\cup \{t\}}(s)=\Vor_{\Omega,S}(s)\cap \Vor_{\Omega,\{s,t\}}(s)$, and insert $\Vor_{\Omega,S}(t)= \bigcap_{s\in S} \Vor_{\Omega,\{s,t\}}(t)$ into the Voronoi diagram. Our software runs in quadratic time in terms of the number of sites and vertices of $\Omega$ and is available at \url{https://github.com/caesardai/Voronoi_In_Hilbert/tree/main/src}.

\section{Concluding Remarks}
In this paper, we presented an analysis of dynamic Voronoi diagrams in the Hilbert metric and presented software for the creation of such diagrams. As discussed in the introduction, the Hilbert metric is useful in the study of quantum information theory \cite{reeb2011hilbertquantum}, real analysis \cite{lemmens2013birkhoff},  optimization \cite{chen2016entropic}, and network analysis \cite{chen2019relaxed}. Research in these fields may benefit from further extensions of algorithmic results in Euclidean to Hilbert geometry.











\newpage

\onecolumn

\end{document}